\documentclass[runningheads]{llncs}

\usepackage[english]{babel}
\usepackage[T1]{fontenc}
\usepackage[utf8]{luainputenc}

\usepackage{algorithmicx}
\usepackage{algpseudocode}

\usepackage{amsmath}
\usepackage{amssymb}
\usepackage{thm-restate}

\usepackage[english]{varioref}
\usepackage[final]{hyperref}
\usepackage[nameinlink]{cleveref}

\usepackage{todonotes}

\usepackage{tikz}
\usetikzlibrary{automata,positioning,arrows}
\tikzset{every state/.style={minimum size=0pt}}

\newcommand{\cupdot}{\mathbin{\mathaccent\cdot\cup}}

\newcommand{\op}[1]{\mathsf{#1}}
\newcommand{\Nat}{\mathbb{N}}
\newcommand{\NFA}{\op{NFA}}
\newcommand{\NBA}{\op{NBA}}
\newcommand{\UBA}{\op{UBA}}

\newcommand{\IDA}{\op{IDA}}
\newcommand{\EDA}{\op{EDA}}
\newcommand{\mc}[1]{\mathcal{#1}}
\newcommand{\Runs}{\op{Runs}}
\newcommand{\Pref}{\op{Pref}}
\newcommand{\Suf}{\op{Suf}}
\newcommand{\Inf}{\op{Inf}}

\newcommand{\da}{\op{da}}
\newcommand{\dpa}{\op{dpa}}

\usepackage{graphicx}
\begin{document}
\title{On finitely ambiguous Büchi automata \thanks{The final authenticated publication is
available online at \texttt{https://doi.org/10.1007/978-3-319-98654-8\_41}}}
%

\author{Christof Löding \and Anton Pirogov 
\thanks{This work is supported by the German research council (DFG) Research Training Group 2236 UnRAVeL}}

\institute{RWTH Aachen University, Templergraben 55, 52062 Aachen, Germany \\
\email{\{loeding,pirogov\}@cs.rwth-aachen.de} }
\maketitle              
\begin{abstract} Unambiguous Büchi automata, i.e. Büchi automata allowing only one
  accepting run per word, are a useful restriction of Büchi automata that is well-suited
  for probabilistic model-checking. In this paper we propose a more permissive variant,
  namely \emph{finitely ambiguous Büchi automata}, a generalisation where each word has at
  most $k$ accepting runs, for some fixed $k$. We adapt existing notions and results
  concerning finite and bounded ambiguity of finite automata to the setting of
  $\omega$-languages and present a translation from arbitrary nondeterministic Büchi
  automata with $n$ states to finitely ambiguous automata with at most $3^n$ states and at
  most $n$ accepting runs per word.

\end{abstract}

\section{Introduction}

Nondeterministic Büchi automata ($\NBA$) \cite{buchi1966symposium} are
finite automata for infinite words that have applications in logical
decision procedures, in particular in the field of model checking
\cite{baier2008principles}, as they can succinctly represent many
interesting properties of non-terminating systems with infinite
execution traces. In some contexts, unrestricted nondeterminism is
problematic, e.g. in probabilistic model checking reasoning about
probabilities becomes very difficult under nondeterminism and
therefore other models are necessary.

One solution is determinisation of Büchi automata. As deterministic Büchi automata are
strictly weaker than $\NBA$, this requires a quite complex translation to automata with
different acceptance conditions and incurs a state blow-up of order $2^{n\log n}$ in the
worst-case \cite{safra1988complexity,thomas1997languages}. But determinisation can be
avoided because some restricted forms of nondeterminism are also suitable for probabilistic model checking, e.g.
\emph{limit-deterministic Büchi automata} \cite{courcoubetis1995complexity}, which can be
separated into a subset of states that can never accept, but have nondeterministic
transitions, and a deterministic subset that contains all of the accepting states, but cannot reach
the nondeterministic states again. Such automata can be used with similar algorithms as
deterministic automata and are also suitable for the model checking of Markov decision
processes \cite{couvreur2003optimal}.

Another well-studied variant are \emph{unambiguous Büchi automata}
($\UBA$), i.e. automata admitting at most one accepting run for each
word, which are known to be as powerful as unrestricted Büchi automata
\cite{arnold1983rational}, while they can be exponentially smaller
than equivalent deterministic automata \cite{baier2016markov}. On finite words, unambiguous
automata form an interesting subclass of nondeterministic automata
because they admit a polynomial time inclusion test
\cite{stearns1985equivalence} (while this problem is
$\op{PSPACE}$-complete for general nondeterministic automata). This
result can be extended to finitely ambiguous automata, which have at
most $k$ accepting runs for each input for some fixed number $k$
\cite{stearns1985equivalence}.  It is unknown whether the polynomial
time inclusion test can be extended to $\UBA$. However, some positive
results have been obtained for simpler types of acceptance conditions
\cite{IsaakL12} and a stronger notion of ambiguity
\cite{bousquet2010equivalence}. Furthermore, $\UBA$ admit a polynomial
time algorithm for quantitative probabilistic model checking based on
linear equation systems \cite{baier2016markov}.

While standard translations from the temporal logic $\op{LTL}$ to $\NBA$
yield unambiguous automata, for the transformation of a given $\NBA$ into
an $\UBA$ only non-trivial constructions
\cite{kahler2008complementation,karmarkar2013improved} roughly of
order $n^n$ for an $n$ state $\NBA$ are known.

In this article, we study finitely ambiguous $\NBA$. To the best of
our knowledge, this model has not been considered before. We show that
there is a simple construction for transforming any given $\NBA$ with
$n$ states into a finitely ambiguous $\NBA$ with at most $3^n$
states. We also present an exponential lower bound of order $2^n$ for such a
construction, which is easily obtained from a corresponding lower
bound for finitely ambiguous automata on finite words
\cite{leung1998separating}. Furthermore, we study the possible degrees
of ambiguity for $\NBA$.  We present a classification of the degree of ambiguity of Büchi automata and the complexity of the corresponding decision
problems, based on results for finite words in
\cite{weber1991degree,allauzen2008general,chan1988finite}. While many results can be
 generalized from finite to infinite words in a straight-forward way, there are different types of infinite degree of ambiguity for $\NBA$. We characterize those in terms of state patterns similar to those that are used over finite words to distinguish between polynomial and exponential growth rates of ambiguity.

This paper is organized as follows. After introducing basic notation
in \Cref{sec:prelim} we present the classification degrees of
ambiguity of nondeterministic Büchi automata and the complexity of the
corresponding decision problems in \Cref{sec:ambiguity}. In
\Cref{sec:translation} we present the translation from an arbitrary
$\NBA$ to a finitely ambiguous $\NBA$, and state a lower bound for such a transformation.
In \Cref{sec:conclusion} we conclude. Full versions of proofs sketched in the main
text can be found in the appendix.
\vspace{-4mm}

\section{Preliminaries}
\vspace{-3mm}
\label{sec:prelim}
For a finite alphabet $\Sigma$, $\Sigma^*$ denotes the set of all finite and
$\Sigma^\omega$ the set of all infinite words over $\Sigma$.
For $a_i\in \Sigma$ and a (finite or infinite) word $w=a_1a_2\dots$ let $w(i):=a_i$. We denote the prefix
of length $i$ with $\Pref_i(w) := a_1\dots a_i$.
The suffix starting at position $i$ is denoted by $\Suf_i(w) := a_i
a_{i+1}\dots$.
Let $\Inf(w) = \{x \mid w(i) = x$ for infinitely many $i\}$ denote the
\emph{infinity set} of a word $w$.

Let a tuple $\mc{A}=(Q,\Sigma,\Delta,Q_0,F)$ denote a \emph{finite automaton} with some
finite alphabet $\Sigma$, finite set of states $Q$, transition relation $\Delta\subseteq
Q\times\Sigma\times Q$ and initial and final states $Q_0,F\subseteq Q$.
Let $|\mc{A}| := |\Delta|$ denote the size of $\mc{A}$.
We write
$\Delta(P,x)$ for $\{q \mid p\in P, (p,x,q)\in\Delta \}$, $\Delta_S(P,x) = \Delta(P,x) \cap S$
and $\Delta_{\overline{S}}(P,x) = \Delta(P,x)\setminus S$ for some $S \subseteq Q$.
For convenience, we write $\Delta(p,x)$ when we mean $\Delta(\{p\},x)$.
With $\mc{A}[S]$ we denote the modified automaton with $Q_0=S$, $\mc{A}[\{s\}]$ can be written as
$\mc{A}[s]$.

A transition sequence $\pi=(q_1,a_1,r_1)\ldots(q_n,a_n,r_n)$ is called a \emph{path} if
$q_i$ equals $r_{i-1}$ for all $1 < i \le n$.
The source and target of the path are denoted by
$\op{src}(\pi)=q_1$ and $\op{trg}(\pi)=r_n$ while $\op{lbl}(\pi)=a_1\dots a_n$ and
$\op{st}(\pi)=q_1\dots q_n r_n$ denote the label and state sequence of $\pi$, respectively.
For convenience, let $\pi(i) := \op{st}(\pi)(i)$. We also consider infinite paths (with
the obvious definition). In general, when speaking of a path or a sequence, we refer to a
finite or infinite path or sequence, depending on the context.

The set of all paths from states in $Q$ to states in $R$ labelled with $x$ is denoted by
$P(Q,x,R) := \{\pi \mid \op{src}(\pi)\in Q, \op{lbl}(\pi)=x, \op{trg}(\pi)\in R \}$, while
$P^\omega(Q,x):=\{\pi \mid \op{src}(\pi)\in Q, \op{lbl}(\pi)=x, x\in\Sigma^\omega\}$ denotes all
infinite paths with label $x$ starting in a state from $Q$. Paths compose in the expected
way. We write $p \overset{x}{\rightarrow} q$ if $P(p,x,q)\neq\emptyset$ and $p \rightarrow
q$ if $p \overset{x}{\rightarrow} q $ for some $x\in \Sigma^*$.
A \emph{strongly connected
component (SCC)} $C\subseteq Q$ of $\mc{A}$ as usual is a maximal (w.r.t. inclusion) subset of states such that if
$p,q \in C$ then $p\rightarrow q$ and $q \rightarrow p$.

The language of $\mc{A}$ when read as $\NFA$ is defined as $L(\mc{A}_\NFA):=\{x\in\Sigma^*
\mid P(Q_0,x,F) \neq \emptyset \}$. The $\omega$-language of $\mc{A}$ when read as
nondeterministic Büchi automaton ($\NBA$) is defined as
$L(\mc{A}_\NBA):=\{x\in\Sigma^\omega\mid
\exists\pi\in P^\omega(Q_0,x), \Inf(\op{st}(\pi)) \cap F \neq \emptyset\}$.
The set $\Runs(\mc{A}, x)$ contains accepting runs of $\mc{A}$ on $x$, i.e. all paths from
an initial state that are labelled with $x$ and satisfy the corresponding
acceptance condition. We say that a set of runs is \emph{separated (at time $i$)}
when the prefixes of length $i$ of those runs are pairwise different.

We say $\mc{A}$ is \emph{trim}, if each path from an initial state is a prefix of an
accepting run and if $\mc{A}$ is an $\NBA$ we additionally require that each accepting
state is on a cycle. This means for $\NFA$ that an accepting
state is always reachable and for $\NBA$ that a cycle with an
accepting state is always reachable and no state is uselessly marked as accepting.
In the following, let $\mc{A}=(Q,\Sigma,\Delta,Q_0,F)$ be some arbitrary finite automaton
if not specified otherwise.

\section{Ambiguity of Büchi automata}
\label{sec:ambiguity}


We first give some basic definitions concerning the degree of
ambiguity of $\NFA$ and $\NBA$. Then we restate some results on the
ambiguity of $\NFA$ given in \cite{weber1991degree}, and continue with the analysis of degrees of ambiguity for $\NBA$.

By $\aleph_0$ we denote the cardinality of the natural numbers and by $2^{\aleph_0}$ the
cardinality of the real numbers.
The \emph{degree of ambiguity} of automaton $\mc{A}$ on a word $x$ is defined as
$\da(\mc{A},x):=|\op{Runs}(\mc{A},x)|$ and the degree of ambiguity of an automaton is given by
$\da(\mc{A}):=\sup_x\{\da(\mc{A},x)\}$ over all possible words $x$. Note that the result
depends on whether we consider $\mc{A}$ to be an $\NFA$ or an $\NBA$ -- in the first case we consider finite input words, in the second case infinite words. If
$\da(\mc{A})<\aleph_0$, $\mc{A}$ is \emph{finitely ambiguous}. We say $\mc{A}$ is
\emph{$k$-ambiguous} if $\da(\mc{A})=k$ and \emph{unambiguous} for $k=1$.
If $\da(\mc{A})\geq \aleph_0$, $\mc{A}$ is \emph{infinitely ambiguous}.

For infinitely ambiguous $\mc{A}_\NFA$ let the \emph{degree of polynomial
ambiguity} $\dpa(\mc{A})$ be the smallest $k\in\Nat$ such that for all $w\in\Sigma^*,
|w|=n, \da(\mc{A},w)\in\mc{O}(n^k)$. $\mc{A}_\NFA$ is \emph{polynomially ambiguous} if
$\dpa(\mc{A})<\infty$, otherwise \emph{exponentially ambiguous}.

So polynomial ambiguity on finite words means that there is no constant upper bound on
the number of accepting runs that holds for all words, but there is a polynomial function
bounding the number of accepting runs for words with a fixed length. Similarly,
exponential ambiguity means that no such polynomial bound exists.
As shown in \cite{weber1991degree}, the different types of ambiguity for
$\NFA$ can be characterized by the following state patterns:

\begin{definition}[Infinite ambiguity conditions for $\NFA$]
  \label{def:idaeda}
  \ \\[-6mm]
  \begin{itemize}
  \item
  $\mc{A}$ satisfies $\IDA$ (infinite degree of ambiguity) if there are $p,q\in Q,p\neq q, v\in\Sigma^*$
  and three paths $\pi_1\in P(p,v,p),\pi_2\in P(p,v,q),\pi_3\in P(q,v,q)$.
  We call a tuple $(p,q,v,\pi_{1,2,3})$ an \emph{$\IDA$ pattern}. 

  \item
  $\mc{A}$ satisfies $\EDA$ (exponential degree of ambiguity) if there is $p\in Q,v\in\Sigma^*$ and two cycles
  $\pi_1,\pi_2\in P(p,v,p),\pi_1\neq \pi_2$.
  We call a tuple $(p,v,\pi_{1,2})$ an \emph{$\EDA$ pattern}. 
  \end{itemize}
  \vspace{-2mm}
  The corresponding paths $\pi_i$ may be omitted when not required.
\end{definition}

\begin{theorem}[\cite{weber1991degree}]
  \label{thm:ambclasses}
  \ \\[-6mm]
  \begin{enumerate}
    \item \label{fact:eda_ida} If $\mc{A}$ satisfies $\EDA$, then $\mc{A}$ satisfies $\IDA$.
    \item
      $\mc{A}$ satisfies $\EDA$ $\Leftrightarrow \mc{A}_\NFA$ is exponentially ambiguous.
    \item
      $\mc{A}$ satisfies $\lnot\EDA$ and $\IDA$ $\Leftrightarrow \mc{A}_\NFA$ is polynomially ambiguous.
    \item
      $\mc{A}$ satisfies $\lnot\IDA$ $\Leftrightarrow \mc{A}_\NFA$ is finitely ambiguous.
  \end{enumerate}
\end{theorem}

A novel aspect when measuring ambiguity of Büchi automata is that there are multiple
degrees of infinite ambiguity when considering infinite words, as a single infinite word
can have infinitely many different accepting runs, which is not possible with finite
words. In fact, for an infinite word the number of accepting runs can even be uncountable.
We will see that for each infinite word the cardinality
of the set of different accepting runs is either finite, equal to $\aleph_0$ or equal to
$2^{\aleph_0}$.

Formally, if there exists $w\in\Sigma^\omega$ with $\da(\mc{A},w)=2^{\aleph_0}$,
i.e. some word $w$ has uncountably many accepting runs, we say that $\mc{A}$ is
\emph{uncountably ambiguous} and we write $\da(\mc{A})=2^{\aleph_0}$.
If $\mc{A}$ is not uncountably ambiguous, but there exists some word $w$ with
$\da(\mc{A},w)=\aleph_0$, i.e.\ $w$ has a countably infinite number of accepting runs,
$\mc{A}_\NBA$ is called \emph{strict-countably ambiguous}.
Later we will show that these ambiguity cases can be characterized using
the following refinements of the state patterns in \Cref{def:idaeda}:

\begin{definition}[Additional ambiguity conditions for $\NBA$]
  \label{def:idafedaf}
  \ \\[-6mm]
  \begin{itemize}
  \item
  $\mc{A}$ satisfies $\IDA_{F}$ if it has an $\IDA$ pattern
  $(p,q,v,\pi_{1,2,3})$ such that $q\in F$.

  \item
  $\mc{A}$ satisfies $\EDA_{F}$ if it has an $\EDA$ pattern
  $(p,v,\pi_{1,2})$ such that $p\in F$.
  \end{itemize}
\end{definition}

If the ambiguity of the $\NBA$ is not finite, but there are also no infinite words with
at least $\aleph_0$ accepting runs, we call the $\NBA$ \emph{limit-countably
ambiguous}. In this case we can adapt the notions of polynomial and exponential ambiguity.
These cases can not be defined in exactly the same way as for $\NFA$ because we only consider words that are
infinite. But we can still preserve the spirit of the definitions of polynomial and
exponential ambiguity by defining them as bounds on the maximal growth of ambiguity
with increasing prefix length of words, instead of whole words.

So formally,
if $\mc{A}_\NBA$ is not finitely ambiguous and not at least strict-countably ambiguous,
it is limit-countably ambiguous. More specifically,
if the function $f$ is an upper bound such that for all $w\in
L(\mc{A}_\NBA)$ and $i\in \Nat$ we have
$|\{\pi\in P(Q_0,\Pref_i(w),Q)\mid \op{Runs}(\mc{A}[\op{trg}(\pi)]_\NBA,
\Suf_{i+1}(w))\neq\emptyset \}| \leq f(i)$, we say that $\mc{A}_\NBA$ it is polynomially
ambiguous if one can choose $f(i) := c\cdot i^d$ with constants $c$ and $d$, and
exponentially ambiguous otherwise.


\begin{figure}[t]
  \begin{center}
    (a)
\begin{tikzpicture}[baseline={([yshift=-.5ex]current bounding box.center)},shorten >=1pt,
  node distance=1cm,inner sep=1pt,on grid,auto]
  \node[state,initial,accepting, initial text={}] (q0)   {$q_0$};
   \node[state] (q1) [right=of q0] {$q_1$};
   \node[state] (q2) [above=of q0] {$q_2$};
   \node[state,accepting] (q3) [above=of q1] {$q_3$};
    \path[->]
    (q0) edge [] node {a} (q1)
    (q0) edge [] node {a,b} (q2)
    (q2) edge [] node [near end] {b} (q1)
    (q1) edge [bend left] node {b,c} (q0)
    (q2) edge [bend left] node {c} (q3)
    (q3) edge [bend left,swap] node {a} (q2)
    ;
\end{tikzpicture}
    \ (b)
\begin{tikzpicture}[baseline={([yshift=-.5ex]current bounding box.center)},shorten >=1pt,
  node distance=1cm,inner sep=1pt,on grid,auto]
  \node[state,initial, initial text={}] (q0)   {$q_0$};
   \node[state] (q1) [right=of q0] {$q_1$};
   \node[state,accepting] (q2) [right=of q1] {$q_2$};
    \path[->]
    (q0) edge [loop above] node {a} (q0)
    (q0) edge [] node {a} (q1)
    (q1) edge [bend left] node {b} (q0)
    (q1) edge [loop above] node {a} (q1)
    (q1) edge [] node {c} (q2)
    (q2) edge [loop above] node {c} (q2)
    ;
\end{tikzpicture}
    \ (c)
\begin{tikzpicture}[baseline={([yshift=-.5ex]current bounding box.center)},shorten >=1pt,
  node distance=1cm,inner sep=1pt,on grid,auto]
  \node[state,initial, initial text={}] (q0)   {$q_0$};
   \node[state] (q1) [right=of q0] {$q_1$};
   \node[state,accepting] (q2) [above=of q1] {$q_2$};
   \node[state,accepting] (q3) [above=of q0] {$q_3$};
    \path[->]
    (q0) edge [loop below] node {a} (q0)
    (q0) edge [] node {a} (q1)
    (q0) edge [] node [] {b} (q2)
    (q1) edge [loop below] node {a} (q1)
    (q1) edge [] node {b} (q2)
    (q0) edge [] node {b} (q3)
    (q3) edge [loop left] node {b} (q3)
    ;
\end{tikzpicture}
    \ (d)
\begin{tikzpicture}[baseline={([yshift=-.5ex]current bounding box.center)},shorten >=1pt,
  node distance=1cm,inner sep=1pt,on grid,auto]
  \node[state,initial,accepting, initial text={}] (q0)   {$q_0$};
   \node[state] (q1) [right=of q0] {$q_1$};
   \node[state] (q2) [above=of q0] {$q_2$};
    \path[->]
    (q0) edge [] node {a} (q1)
    (q0) edge [] node {a,b} (q2)
    (q2) edge [] node [near end] {b} (q1)
    (q1) edge [bend left] node {b} (q0)
    ;
\end{tikzpicture}
  \end{center}
  \caption{
    (a) The word $(ab)^\omega$ is accepted unambiguously, while $ab^\omega$ has two accepting
    runs due to the choice of the first transition. The word $(ac)^\omega$ has
    strict-countable ambiguity as $q_0 \overset{ac}{\rightarrow} q_0,
    q_0\overset{ac}{\rightarrow} q_3$ and $q_3 \overset{ac}{\rightarrow} q_3$,
    which is an $\IDA_F$ pattern. The word $(acabb)^\omega$ has uncountably many accepting
    runs due to two paths $q_0\overset{acabb}{\rightarrow} q_0$, implying $\EDA_F$ and
    therefore the automaton is uncountably ambiguous.
    (b) This automaton has an $\IDA$ pattern $(q_0,q_1,a)$ and an $\EDA$ pattern $(q_0,
    aab)$, but no $\IDA_F$ nor $\EDA_F$. The word $a^*ac^\omega$ is
    polynomially ambiguous and $(aab)^*ac^\omega$ is exponentially ambiguous, as the
    corresponding pattern can be traversed in different ways a finite number of times,
    before reading the first $c$.
    (c) Counterexample of \Cref{lem:nfanba} for non-trim automata: $L(\mc{A}_\NFA)=a^*
    b^+$ and $\mc{A}_\NFA$ is not finitely ambiguous, because each word $a^n b$ has for
    $0<i\leq n$ the accepting runs $q_0^i q_1^{n-i}q_2$, while
    $L(\mc{A}_\NBA)=a^*b^\omega$ is unambiguous, as the only accepting path must have a
    state sequence of the form $q_0^*q_3^\omega$.
    (d) As $\NFA$ this automaton is unambiguous, while as $\NBA$ there are two accepting
    runs on $ab^\omega$.
  }

  \label{fig:ambig}
\end{figure}
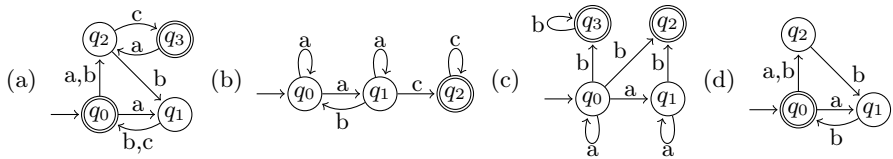

\begin{figure}
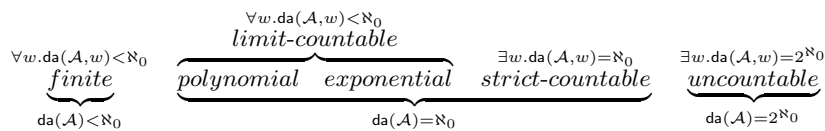

  \begin{center}
    \(
    \overset{\forall w.\da(\mc{A},w)<\aleph_0}{
      \underbrace{finite}_{\da(\mc{A}) < \aleph_0}
    }
    \quad
    \underbrace{
      \overbrace{polynomial \quad exponential}^{
        \overset{\scriptstyle{\forall w.\da(\mc{A},w)<\aleph_0}}{\textstyle{limit\text{-}countable}}
      }
      \quad
      \overset{\exists w.\da(\mc{A},w)=\aleph_0}{strict\text{-}countable}
    }_{\da(\mc{A}) = \aleph_0}
    \quad
    \overset{{\exists w.\da(\mc{A},w)=2^{\aleph_0}}}{
      \underbrace{uncountable}_{\da(\mc{A}) = 2^{\aleph_0}}
    }
    \)
  \end{center}
  \caption{ Illustration of the ambiguity hierarchy for $\NBA$.
  The five classes (without "limit-countable") are pairwise disjoint.
  The depicted order reflects the meaning of e.g. "at most polynomial amb.", which includes
  finite and polynomial ambiguity, or "at least strict-countable amb.", which includes
  strict-countable and uncountable ambiguity.}
  \label{fig:hierarchy}
\end{figure}

Consider \Cref{fig:ambig} (a,b) to see examples of different ambiguity types on infinite words.
In the following we will justify the ambiguity claims in the examples,
by relating the patterns from \Cref{def:idafedaf} to the various cases that emerge for
$\NBA$. The resulting hierarchy is illustrated in \Cref{fig:hierarchy} and
summarized in \Cref{thm:newambclasses}.
We start by establishing that the notion of finite ambiguity of an automaton is closely
related for $\NFA$ and $\NBA$ under the condition that the considered automaton must be
trim as $\NBA$.

\begin{lemma}
  \label{lem:nfanba} Let $\mc{A}$ be trim as $\NBA$. Then $\mc{A}_\NFA$ is finitely
  ambiguous if and only if $\mc{A}_\NBA$ is finitely ambiguous.
\end{lemma}
\begin{proof}
  For the first claim, assume $\mc{A}$ is not a finitely ambiguous $\NBA$. Then for all $k\in\Nat$
  there is a word with at least $k$ different runs.
  Pick a word $w\in\Sigma^\omega$ with at least $k|Q|$ different runs and a time where all
  these runs are separated. Then there are at least $k$ runs that are in the same
  state $p\in Q$. We can extend these prefixes by a path from $p$ to some $q\in F$,
  obtaining a word that is accepted by at least $k$ runs of the corresponding $\NFA$.

  The second claim is similar, with the difference that $k$ different finite runs on
  some finite word that end in the same state are extended to accepting infinite runs (this
  requires that $\mc{A}$ is trim).
  \qed
\end{proof}

See \Cref{fig:ambig} (c) for a non-trim counterexample and also notice
that in (d) the automaton has a different finite ambiguity as $\NFA$ than as $\NBA$.
Hence, in general $\da(\mc{A}_\NBA) \neq \da(\mc{A}_\NFA)$ and the calculation of the
exact degree must also be adapted to the $\NBA$ setting. We will sketch a corresponding
procedure later in the context of \Cref{thm:ambpspace}.

We now state some technical lemmas for establishing the connection between the state patterns and the degrees of ambiguity summarized in \Cref{thm:newambclasses} further below.

\begin{restatable}{lemma}{lemedafimpliesidaf}
  \label{lem:edaf_implies_idaf}
    If $\mc{A}$ satisfies $\EDA_{F}$, then $\mc{A}$ satisfies $\IDA_{F}$.
\end{restatable}
\vspace{-4mm}
\begin{proof}[sketch]
  Similar argument as for the implication $\EDA \Rightarrow \IDA$ in \cite{weber1991degree}.
  \qed
\end{proof}

\begin{restatable}{lemma}{lemnoedaboundamb}
  \label{lem:noeda_boundamb}
  If $\mc{A}$ satisfies $\lnot\EDA_F$, then for all $q\in F, w\in \Sigma^\omega$ the number of
  infinite paths of $\mc{A}[q]$ visiting $q$ infinitely often is at most $|Q|$.
\end{restatable}
\vspace{-4mm}
\begin{proof}[sketch]
  Shown by simple construction of an $\EDA_F$ pattern in case of more than $|Q|$ such paths.
  \qed
\end{proof}

Now we can relate the extended patterns to the new infinite ambiguity cases.

\begin{lemma}
  \label{lem:edafunc}
    $\mc{A}$ satisfies $\EDA_F \Leftrightarrow \mc{A}_\NBA$ is uncountably ambiguous.
    Furthermore, if $\mc{A}$ does not satisfy $\EDA_F$, then $\mc{A}_\NBA$ is at most
    strict-countably ambiguous.
\end{lemma}
\begin{proof}
  For one direction,
  let $(p,v,\pi_{1,2})$ be an $\EDA$ pattern satisfying $\EDA_F$.
  Pick some $u\in \Sigma^*, \pi_0\in P(Q_0,u,p)$. Clearly, $uv^\omega \in L(\mc{A}_\NBA)$.
  Observe that for $v^\omega =v_0v_1\ldots$ each $v_i$ can be consumed by taking either $\pi_1$
  or $\pi_2$. Hence, the number of runs on $uv^\omega$ is uncountable
  , i.e., $\da(\mc{A}_\NBA)=\da(\mc{A}_\NBA,
  uv^\omega) = 2^{\aleph_0}$.

  For the other direction, let $w\in L(\mc{A}_\NBA)$ and assume $\EDA_F$ does not hold.
  For each $i\in\Nat$, let $x=\Pref_i(w)$ and notice that the number of different
  paths $P_x=P(Q_0,x,q)$ must be finite for each $q\in F$, as $x$ is finite. By
  \Cref{lem:noeda_boundamb}, for each such path there are at most $|Q|$ continuations to
  infinite runs that visit $q$ infinitely often. It follows that $\mc{A}_\NBA$ is at most
  strict-countably ambiguous.
  \qed
\end{proof}

\begin{restatable}{lemma}{lemidafc}
  \label{lem:idafc}
    $\mc{A}$ satisfies IDA$_F$ $\Leftrightarrow \mc{A}_\NBA$
    is at least strict-countably ambiguous.
\end{restatable}
\begin{proof}[sketch]
  For one direction,
  let $(p,q,v,\pi_{1,2,3})$ be an $\IDA$ pattern satisfying $\IDA_F$.
  Let $u\in \Sigma^*,\pi_0\in P(Q_0,u,p)$. Clearly, $uv^\omega\in L(\mc{A}_\NBA)$.
  Observe that for each $i\in \Nat$, $\pi_1$ can be taken $i$ times
  before using $\pi_2$ and then taking path $\pi_3$ forever. Hence,
  $\da(\mc{A}_\NBA) \geq \da(\mc{A}_\NBA, uv^\omega) \geq \aleph_0$.

  For the other direction, let $w\in L(\mc{A}_\NBA)$ and assume $\IDA_F$ does not hold. We show that the number of runs on $w$ is finite, and thus the degree of ambiguity of $\mc{A}_\NBA$ is at most limit-countable. Since $\IDA_F$ does not hold,
  $\EDA_F$ does not hold by \Cref{lem:edaf_implies_idaf}. Then by
  \Cref{lem:noeda_boundamb} a run can separate into at most $|Q|$ different accepting runs
  that visit $q\in F$ infinitely often, after seeing $q$ the first time.
  Assume that there is some state $q\in F$ such that there are infinitely many accepting
  runs that visit $q$ infinitely often. Then there must be such runs for which the first visit to $q$ happens arbitrarily late. From that observation one can construct an $\IDA_F$ pattern, a contradiction.
  But then the number of accepting runs on $w$ must be finite.
  \qed
\end{proof}

The relationship of the different state patterns in $\mc{A}$ and the corresponding
ambiguity classes of the Büchi automaton are summed up in the following theorem:

\begin{theorem}
  \label{thm:newambclasses}
  Let $\mc{A}$ be a trim $\NBA$.
  \ \\[-5mm]
  \begin{enumerate}
    \item
      $\mc{A}$ satisfies $\EDA_F$ $\Leftrightarrow \mc{A}_\NBA$ is uncountably ambiguous.
    \item
      $\mc{A}$ satisfies $\lnot\EDA_F$ and $\IDA_F$ $\Leftrightarrow \mc{A}_\NBA$ is strict-countably ambiguous.
    \item
      $\mc{A}$ satisfies $\lnot\IDA_F$ and $\EDA$ $\Leftrightarrow \mc{A}_\NBA$ is exponentially ambiguous.
    \item
      $\mc{A}$ satisfies $\lnot\IDA_F$,$\lnot\EDA$ and $\IDA$ $\Leftrightarrow \mc{A}_\NBA$ is polynomially ambiguous.
    \item
      $\mc{A}$ satisfies $\lnot\IDA_F$ and $\IDA$ $\Leftrightarrow \mc{A}_\NBA$ is limit-countably ambiguous.
    \item
      $\mc{A}$ satisfies $\lnot\IDA$ $\Leftrightarrow \mc{A}_\NBA$ is finitely ambiguous.
  \end{enumerate}
\end{theorem}
\begin{proof}
  \ \\[-5mm]
  \begin{enumerate}
    \item[(1+2):] By \Cref{lem:edafunc} and \Cref{lem:edafunc,lem:idafc}, respectively.
    \item[(3+4):] We show both directions by establishing that the difference in growth of
      ambiguity of $\mc{A}_\NFA$ in the length of finite words and the growth of
      ambiguity of $\mc{A}_\NBA$ in the length of finite prefixes is bounded by constants.

      $(\Rightarrow):$
      $\lnot\IDA_F$ implies $\lnot\EDA_F$ by \Cref{lem:edaf_implies_idaf} and by
      \Cref{lem:idafc} $\mc{A}_\NBA$ is not strict-countably ambiguous, hence no word has
      infinitely many runs. But by \Cref{thm:ambclasses} \cite{weber1991degree} there is a
      family of finite words with increasing length that w.l.o.g. terminate in the same
      $q\in F$ and witness the polynomial (exponential) ambiguity of $\mc{A}_\NFA$. As
      $\mc{A}$ is trim, for each such word $u$ with $k$ accepting runs, those runs can be
      extended by the same loop from $q$ to $q$ labelled with some $v$, hence $uv^\omega$
      has at least $k$ (and by \Cref{lem:noeda_boundamb} at most $k|Q|$) accepting runs in
      $\mc{A}$ as $\NBA$. Hence $\mc{A}_\NBA$ has asymptotically the same limit-countable
      ambiguity.

      $(\Leftarrow):$
      Pick some $w\in L(\mc{A}_\NBA)$ with $\da(\mc{A}_\NBA,w)=k$
      and pick $i$ such that after reading the prefix of length $i$ all accepting runs
      have separated. Notice that there must be at least $\lfloor \frac{k}{|Q|} \rfloor$
      different runs that are in the same state $p$. As $\mc{A}$ is trim, there is some
      finite word $x$ with $|x|\leq |Q|$ leading from $p$ to some $q\in F$. Let $a
      \leq |Q|^{|Q|}$ be the maximum number of different paths of length $|Q|$ in $\mc{A}$
      that have the same source, target and label and let $\hat{w} = \Pref_i(w)x$. Then on
      any $x$ each run can separate into at most $a$ different runs and therefore we
      have that $\lfloor \frac{k}{|Q|} \rfloor \leq \da(\mc{A}_\NFA,\hat{w}) \leq ak$.
      Hence, by picking for each $i$ an infinite word $w$ with the largest
      number of  separated accepting runs after reading $\Pref_i(w)$, we can construct a
      family of finite words with the number of accepting runs growing asymptotically in
      the same way as the maximum number of separated accepting runs grows on prefixes of
      infinite words. By \Cref{thm:ambclasses} \cite{weber1991degree} this implies that
      $\mc{A}$ must satisfy $\IDA$ or $\EDA$, respectively.


    \item[(5):] By $(3),(4)$ and \Cref{thm:ambclasses},
      as $\EDA$ implies $\IDA$.
    \item[(6):] Shown in \Cref{thm:ambclasses} for trim $\NFA$.
      By \Cref{lem:nfanba} this extends to trim $\NBA$.
  \qed
  \end{enumerate}
\end{proof}

The ambiguity class of an $\NBA$ from the hierarchy in \Cref{thm:newambclasses}
can also be determined efficiently:

\begin{theorem}
  \ \\[-5mm]
  \begin{enumerate}
    \item Uncountable ambiguity of an $\NBA$ $\mc{A}$ can be decided in $\mc{O}(|\mc{A}|^2)$.
    \item The ambiguity class of an $\NBA$ $\mc{A}$ can be computed in
      $\mc{O}(|\mc{A}|^3)$.
    \item $\dpa(\mc{A}_\NBA)$ can be computed in $\mc{O}(|\mc{A}|^3)$.
  \end{enumerate}
\end{theorem}
\begin{proof}
  A straightforward modification of the corresponding algorithms for $\NFA$ from
  \cite{allauzen2008general}, which use a depth-first search in products of $\mc{A}$ with
  itself. $\EDA_F$ and $\IDA_F$ can easily be checked by adding the corresponding
  restriction to the $\IDA$ / $\EDA$ pattern, i.e., that a specific state must
  additionally be accepting. In case of polynomial ambiguity, the algorithm to calculate
  $\dpa(\mc{A}_\NBA)$ can be used without changes after excluding $\EDA$ and $\IDA_F$.
  \qed
\end{proof}


Computing the exact degree of finite ambiguity, similar to the case for $\NFA$ \cite{chan1988finite}, is a difficult
problem:
\begin{restatable}{theorem}{thmambpspace}
  \label{thm:ambpspace}
    Deciding whether $\da(\mc{A}_\NBA)>d$ for a given automaton $\mc{A}$ and $d\in\Nat$ is a $\op{PSPACE}$-complete problem.
\end{restatable}
\begin{proof}[sketch]
  We adapt the algorithm presented in \cite{chan1988finite} to show
  $\op{PSPACE}$-completeness of this problem from $\NFA$ to $\NBA$. The nondeterministic
  algorithm for $\NFA$ guesses a word with at least $d+1$ accepting runs and evaluates the
  product of the transition matrices for each symbol along the word (which yields the
  number of different runs), while bounding the growth of the numbers. First we show that
  we can restrict ourselves to ultimately periodic words of the form $uv^\omega$ with
  $u,v\in \Sigma^*$ and then we argue that it suffices to guess $u$ and $v$ accordingly.
  We show completeness by a simple reduction from the $\NFA$ to the $\NBA$ problem that
  maps each finite word $w$ accepted by the $\NFA$ one-to-one to a word $w\#^\omega$ accepted
  by the $\NBA$ with the same number of accepting runs. \qed
\end{proof}

\section{Translation from NBA to finitely ambiguous NBA}
\label{sec:translation}
\newcommand{\A}{\mathcal{A}}
\newcommand{\rst}{T_{\mathrm{rs}}^{\A,w}}
\newcommand{\st}{T_{\mathrm{s}}^{\A,w}}
We present a construction that converts a given $\NBA$ with $n$ states into a finitely
ambiguous $\NBA$ with degree of ambiguity at most $n$. For explaining the intuition of the
construction and for proving its correctness, we first introduce in
\Cref{subsec:split-trees} the notion of the reduced split tree as defined in \cite{kahler2008complementation}. This tree, defined for an $\NBA$ $\A$ and an infinite word $w$, collects runs of $\A$ on $w$ in a specific way.

The translation of an $\NBA$ into a finitely ambiguous $\NBA$ is presented in \Cref{subsec:faba-construction}. The construction uses two disjoint sets for tracking runs of the given $\NBA$, and its description does not rely on the notion of reduced split tree. However, for understanding the role of the two subsets in the construction, reduced split trees are a valuable tool.

\subsection{Reduced Split Trees} \label{subsec:split-trees}
%
%
An \emph{$X$-labelled binary tree} is a partial function  $T: \{0,1\}^* \rightarrow X$ such that the domain $N_T$ of $T$ is a non-empty prefix-closed set. The elements of $N_T$ are called the nodes of $T$. The root node is  $\varepsilon$ (since $N_T$ is prefix closed, $\varepsilon \in N_T$). For a node $u \in \{0,1\}^*$, the node $u0$ is the \emph{left child} of $u$, and $u1$ is the \emph{right child} of $u$. For two nodes $u,v \in \{0,1\}^*$ we say that $u$ and $v$ are on the same level if $|u|=|v|$, and we further say that $u$ is to the left of $v$ if $u$ is lexicographically smaller than $v$.

An \emph{infinite path} through such a tree corresponds to an infinite sequence $\pi \in \{0,1\}^\omega$ (the nodes on the path are the finite prefixes of $\pi$). We say that $\pi$ is \emph{left-recurring} if it contains infinitely many $0$ (which means that it moves to the left successor infinitely often).




Let $\A=(Q,\Sigma,\Delta,Q_0,F)$ be an $\NBA$ and $w \in \Sigma^\omega$ be an infinite word. The following definitions are illustrated by simple examples in \Cref{fig:trees}(a)--(c).

The \emph{split tree} $\st$  is a $2^Q$-labelled binary tree with node set $N_T = \{0,1\}^*$ defined as follows: The root is labelled
with the set of initial states $\st(\varepsilon) := Q_0$. If $u\in \{0,1\}^*$ with $|u|=i$ is labelled with $\st(u)
= P\subseteq Q$, then $\st(u0) := \Delta_F(P,w(i))$ and
$\st(u1) := \Delta_{\overline{F}}(P,w(i))$.

Observe that $\st$ is an infinite complete binary tree that encodes all
runs of $\mathcal{A}$ on $w$. It groups the runs by their visits to the set $F$ by continuing runs that pass through an accepting state to the left, and the other ones to the right.
Note that $\st$ can have nodes with label $\emptyset$, and that a state can occur in many labels on each level.


The \emph{reduced left-right tree} $\rst$ is obtained from
$\st$ by keeping for each state only the leftmost occurrence on each level (i.e., if $u$ is to the left of $v$, and $q$ occurs in the labels of $u$ and $v$, then it is removed from the label of $v$) and then removing vertices
labelled with $\emptyset$. This results in a tree because if the label of a node becomes empty by the above operation, then also the labels of its successors become empty.

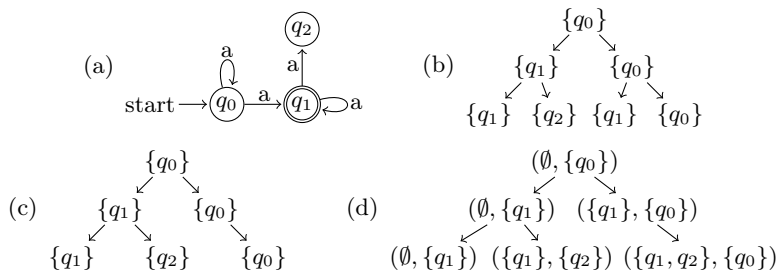
\begin{figure}[t]
  \begin{center}
    (a)
\begin{tikzpicture}[baseline={([yshift=-.5ex]current bounding box.center)},shorten >=1pt,
  node distance=1cm,inner sep=1pt,on grid,auto]
   \node[state,initial]   (q0)               {$q_0$};
   \node[state,accepting] (q1) [right=of q0] {$q_1$};
   \node[state]           (q2) [above=of q1] {$q_2$};
    \path[->]
    (q0) edge [loop above] node {a} (q0)
    (q1) edge [loop right] node {a} (q1)
    (q0) edge node {a} (q1)
    (q1) edge node {a} (q2)
    ;
\end{tikzpicture}
    \quad\quad (b)
\begin{tikzpicture}[baseline={([yshift=-.5ex]current bounding box.center)},shorten >=1pt,
  node distance=0.9cm,on grid,auto,inner sep=1pt]
  \node   (q)               {$\{q_0\}$};
  \node [below left=of q]  (q0)               {$\{q_1\}$};
  \node [below right=of q]  (q1)               {$\{q_0\}$};
  \node [below left=of q0]  (q00)               {$\{q_1\}$};
  \node [below right=of q0,xshift=-4mm]  (q01)               {$\{q_2\}$};
  \node [below left=of q1,xshift=4mm]  (q10)               {$\{q_1\}$};
  \node [below right=of q1]  (q11)               {$\{q_0\}$};
    \path[->]
    (q) edge node {} (q0)
    (q) edge node {} (q1)
    (q0) edge node {} (q00)
    (q0) edge node {} (q01)
    (q1) edge node {} (q10)
    (q1) edge node {} (q11)
    ;
\end{tikzpicture}

    \vspace{2mm}
    (c)
\begin{tikzpicture}[baseline={([yshift=-.5ex]current bounding box.center)},shorten >=1pt,
  node distance=0.9cm,on grid,auto,inner sep=1pt]
  \node   (q)               {$\{q_0\}$};
  \node [below left=of q]  (q0)               {$\{q_1\}$};
  \node [below right=of q]  (q1)               {$\{q_0\}$};
  \node [below left=of q0]  (q00)               {$\{q_1\}$};
  \node [below right=of q0]  (q01)               {$\{q_2\}$};
  \node [below right=of q1]  (q11)               {$\{q_0\}$};
    \path[->]
    (q) edge node {} (q0)
    (q) edge node {} (q1)
    (q0) edge node {} (q00)
    (q0) edge node {} (q01)
    (q1) edge node {} (q11)
    ;
\end{tikzpicture}
    \quad\quad (d)
\begin{tikzpicture}[baseline={([yshift=-.5ex]current bounding box.center)},shorten >=1pt,
  node distance=0.9cm,on grid,auto,inner sep=1pt]
  \node   (q)               {$(\emptyset,\{q_0\})$};
  \node [below left=of q,xshift=-2mm]  (q0)               {$(\emptyset,\{q_1\})$};
  \node [below right=of q,xshift=2mm]  (q1)               {$(\{q_1\},\{q_0\})$};
  \node [below left=of q0,xshift=-4mm]  (q00)               {$(\emptyset,\{q_1\})$};
  \node [below right=of q0,xshift=-1mm]  (q01)               {$(\{q_1\},\{q_2\})$};
  \node [below right=of q1,xshift=2mm]  (q11)               {$(\{q_1,q_2\},\{q_0\})$};
    \path[->]
    (q) edge node {} (q0)
    (q) edge node {} (q1)
    (q0) edge node {} (q00)
    (q0) edge node {} (q01)
    (q1) edge node {} (q11)
    ;
\end{tikzpicture}
  \end{center}

  \caption{ (a) An $\NBA$ $\mc{A}$
  (b) First levels of $\st$ for $w= a^\omega$
  (c) First levels of $\rst$ for $w= a^\omega$
  (d) First two steps of possible transitions of $\A'$
  }
  \label{fig:trees}
\end{figure}

The interesting properties of $\rst$ are summarized in the following lemma, which is shown in \cite{kahler2008complementation}.

\begin{lemma}[\cite{kahler2008complementation}]
  \label{lem:reduced-split-tree}
  \begin{enumerate}
  \item
    The word $w$ is accepted by $\A$ if, and only if, $\rst$ contains an infinite left-recurring path.
  \item $\rst$ contains at most $|Q|$ infinite paths.
  \end{enumerate}

\end{lemma}
The second claim is obvious because each state of $\A$ is contained in at most one label of each level. The first claim is shown by picking a run of $\A$ on $w$ that always visits the next accepting states as early as possible (details can be found in \cite{kahler2008complementation}).

\subsection{Construction of a Finitely Ambiguous NBA} \label{subsec:faba-construction}
We construct from $\A$ a new $\NBA$ $\A'$ whose infinite runs on a word $w$ are in one-to-one correspondence with the infinite paths of $\rst$. The run moves into a final state if the corresponding path branches to the left. Then \Cref{lem:reduced-split-tree} implies that $\A'$ is equivalent to $\A$, and its degree of ambiguity is bounded by $|Q|$.

In order to implement this idea, $\A'$ guesses an infinite path through $\rst$ by constructing its label sequence. If the current label is $S$ (corresponding to some node $u$ in $\rst$), then the labels $S_0$ and $S_1$ of $u0$ and $u1$ in $\rst$ can be constructed based on the next input letter, the transition relation of $\A$, and the knowledge which states occur in the labels to the left of $u$ in $\rst$.
This latter information is tracked in a second set $P$. Thus, in $\A'$ a state $(P,S)$ for $P,S \subseteq Q$ is reachable by reading the first $i$ letters of an input word $w$ if $S$ is the label of some node $u$ on level $i$ in $\rst$, and $P$ is the union of the labels of nodes to the left of $u$ (on the same level).

Formally, we define the new automaton in the following way.
Given $\NBA\ \mc{A}$, let $\mathcal{A'} = (Q', \Sigma, Q_0', \Delta', F')$ be
defined by
  \begin{itemize}
    \item $Q' = \{(P,S) \in 2^Q \times 2^Q \mid P\cap S = \emptyset \mbox{ and }S \not=\emptyset\}$
    \item $Q_0' = (\emptyset, Q_0)$
    \item $\Delta' = \Delta_0' \cup \Delta_1'$ with
      \begin{itemize}
        \item $\Delta_0' = \{((P,S), a, (P',S')) \mid P' = \Delta(P,a),
                S'=\Delta_F(S,a)\setminus P' \}$,
        \item $\Delta_1' = \{((P,S), a, (P',S')) \mid P' = \Delta(P,a) \cup \Delta_F(S,a),
                S' =\Delta_{\overline{F}}(S,a)\setminus P' \}$
      \end{itemize}
  Note that the transitions are only defined if $S' \not= \emptyset$ by definition of $Q'$.
    \item $F' = \{(P,S) \mid S \subseteq F \}$
  \end{itemize}


\Cref{fig:trees}(d) illustrates the possible transitions of $\A'$.

The following lemma is a direct consequence of the definition of $\rst$ and the construction of $\A'$ (by an induction on $i$).
\begin{lemma} \label{lem:automaton-rst}
  Let $x$ be the finite prefix of length $i$ of $w$, and assume that $\A'$ can reach the state $(P,S)$ by reading $x$ from its initial state. Then there is a node $u$ on level $i$ of $\rst$ with label $S$, and $P$ is the union of the labels of the nodes to the left of $u$ on level $i$.
\end{lemma}

\Cref{lem:automaton-rst} implies that the infinite runs of $\A'$ on $w$ are in one-to-one correspondence with the infinite paths of $\rst$ (and therefore $\A'$ has degree of ambiguity at most $|Q|$). The final states of $\A'$ are those in which the second component is a subset of $F$, and therefore correspond to a left successor on the corresponding path in $\rst$. Hence, $\A'$ has an accepting run of $w$ if, and only if, $\rst$ has a left-recurring path.
By \Cref{lem:reduced-split-tree}, this implies that $L(\A) = L(\A')$.

The number of disjoint pairs of subsets of $Q$ is not larger than $3^{|Q|}$. In summary, we obtain the following result.
\begin{theorem}
  Let $\mathcal{A}$ be an $\NBA$ with $n$ states. Then there exists an automaton
  $\mc{A}'$ with at most $3^n$ states accepting the same language such that
  $\op{da}(\mc{A}') \leq n$.
\end{theorem}

An exponential lower bound for the construction of finitely ambiguous $\NBA$ can be inferred from a corresponding lower bound for $\NFA$.

\begin{restatable}{theorem}{thmfabalowerbound}
  \label{thm:fabalowerbound}
  For each $n\in\Nat$ there exists an $\NBA$ with $n$ states such that each finitely
  ambiguous $\NBA$ accepting the same language has at least $2^{n}-1$ states.
\end{restatable}
\begin{proof}[sketch]
  In \cite{leung1998separating} a family $\{\A_i\}_{i\in\Nat}$ of $\NFA$ is presented such that each
  $\A_n$ has $n$ states, and each at most polynomially ambiguous $\NFA$ equivalent to $\A_n$ has
  $2^{n}-1$ states. So the lower
  bound also holds for finitely ambiguous $\NFA$. This lower bound can be lifted to $\NBA$
  by considering the languages $L_n' = (L(\A_n)\#)^\omega$, for which one easily obtains
  $n$ state $\NBA$ from the $\NFA$ $\A_n$. Furthermore, from a finitely ambiguous $\NBA$ for
  $L_n'$ one can extract a finitely ambiguous $\NFA$ for $L(\A_n)$ with the same number of
  states. \qed
\end{proof}

Our construction is already close to this bound. As it is tracking just two subsets, it is also much simpler than the translations used to obtain unambiguous automata
presented in \cite{kahler2008complementation,karmarkar2013improved} that have upper bounds of $4(3n)^n$ and $n(0.76n)^n$ respectively, and are obtained at the cost of much
more involved constructions.
To the best of our knowledge, it is not known whether there is a stronger lower bound for
the construction of unambiguous $\NBA$ than the one in \Cref{thm:fabalowerbound}.



\section{Conclusion}
\label{sec:conclusion}

In this paper we presented a fine classification for the ambiguity of nondeterministic
Büchi automata by lifting results known for $\NFA$ and extending them to precisely capture
the subtle differences in the case of infinite ambiguity. Finally we presented and
discussed a translation from $\NBA$ to finitely ambiguous $\NBA$. In future work we plan
to investigate how this partial disambiguation can be applied in the setting of
probabilistic model checking and look for cases in which using finitely ambiguous $\NBA$
could have an advantage over full disambiguation or determinisation.

%
\bibliographystyle{splncs04}
\bibliography{literature}

\begin{thebibliography}{10}
\providecommand{\url}[1]{\texttt{#1}}
\providecommand{\urlprefix}{URL }
\providecommand{\doi}[1]{https://doi.org/#1}

\bibitem{allauzen2008general}
Allauzen, C., Mohri, M., Rastogi, A.: General algorithms for testing the
  ambiguity of finite automata. In: DLT 2008. pp. 108--120. Springer

\bibitem{arnold1983rational}
Arnold, A.: Rational $\omega$-languages are non-ambiguous. Theoretical Computer
  Science  \textbf{26}(1-2),  221--223 (1983)

\bibitem{baier2008principles}
Baier, C., Katoen, J.: Principles of model checking. {MIT} Press (2008)

\bibitem{baier2016markov}
Baier, C., Kiefer, S., Klein, J., Kl{\"u}ppelholz, S., M{\"u}ller, D., Worrell,
  J.: {Markov} chains and unambiguous {B{\"u}chi} automata. In: CAV 2016. pp.
  23--42. Springer

\bibitem{bousquet2010equivalence}
Bousquet, N., L{\"{o}}ding, C.: Equivalence and inclusion problem for
  {Strongly} {Unambiguous} {B{\"{u}}chi} {Automata}. In: {LATA} 2010. pp.
  118--129

\bibitem{buchi1966symposium}
B{\"u}chi, J.R.: On a decision method in restricted second order arithmetic.
  In: Studies in Logic and the Foundations of Mathematics, vol.~44, pp. 1--11.
  Elsevier (1966)

\bibitem{chan1988finite}
Chan, T.h., Ibarra, O.H.: On the finite-valuedness problem for sequential
  machines. Theoretical Computer Science  \textbf{23}(1),  95--101 (1988)

\bibitem{courcoubetis1995complexity}
Courcoubetis, C., Yannakakis, M.: The complexity of probabilistic verification.
  Journal of the ACM  \textbf{42}(4),  857--907 (1995)

\bibitem{couvreur2003optimal}
Couvreur, J.M., Saheb, N., Sutre, G.: An optimal automata approach to {LTL}
  model checking of probabilistic systems. In: LPAR 2003. pp. 361--375.
  Springer

\bibitem{IsaakL12}
Isaak, D., L{\"{o}}ding, C.: Efficient inclusion testing for simple classes of
  unambiguous {\(\omega\)}-automata. Inf. Process. Lett.  \textbf{112}(14-15),
  578--582 (2012)

\bibitem{kahler2008complementation}
K{\"a}hler, D., Wilke, T.: Complementation, disambiguation, and determinization
  of {B{\"u}chi} automata unified. In: ICALP 2008. pp. 724--735. Springer

\bibitem{karmarkar2013improved}
Karmarkar, H., Joglekar, M., Chakraborty, S.: Improved upper and lower bounds
  for {B{\"u}chi} disambiguation. In: ATVA 2013. pp. 40--54. Springer

\bibitem{leung1998separating}
Leung, H.: Separating exponentially ambiguous finite automata from polynomially
  ambiguous finite automata. SIAM Journal on Computing  \textbf{27}(4),
  1073--1082 (1998)

\bibitem{safra1988complexity}
Safra, S.: On the complexity of omega-automata. In: Foundations of Computer
  Science, 1988., 29th Annual Symposium on. pp. 319--327. IEEE (1988)

\bibitem{stearns1985equivalence}
Stearns, R.E., Hunt~III, H.B.: On the equivalence and containment problems for
  unambiguous regular expressions, regular grammars and finite automata. SIAM
  Journal on Computing  \textbf{14}(3),  598--611 (1985)

\bibitem{thomas1997languages}
Thomas, W.: Languages, automata, and logic. In: Handbook of formal languages,
  pp. 389--455. Springer (1997)

\bibitem{weber1991degree}
Weber, A., Seidl, H.: On the degree of ambiguity of finite automata.
  Theoretical Computer Science  \textbf{88}(2),  325--349 (1991)

\end{thebibliography}

\newpage
\appendix
\section{Full proofs}

In this appendix we provide the full proofs that were omitted in the main paper.

\subsection{Proofs for \Cref{thm:newambclasses}}

Here we provide the full proofs for the technical lemmas that are needed to obtain the
ambiguity hierarchy for $\NBA$ stated in \Cref{thm:newambclasses}.


\begin{lemma}[Pattern shifting]
  \label{lem:shifting}
  \ \\[-6mm] \begin{enumerate}
  \item
  Iff $\mc{A}$ has an $\IDA$ pattern $(p,q,v,\pi_{1,2,3})$ visiting some $q\in F$ on
      $\pi_3$, \\then $\IDA_F$ holds.
  \item
  Iff $\mc{A}$ has an $\EDA$ pattern $(p,v,\pi_{1,2})$ visiting some $q\in F$ on
      $\pi_1$ or $\pi_2$, \\then $\EDA_F$ holds.
  \end{enumerate}
\end{lemma}
\begin{proof}
  $(\Leftarrow)$ holds by definition of $\IDA_F$ and $\EDA_F$.

  $(\Rightarrow)$ for $(1)$:
  Let $(p,q,v,\pi_{1,2,3})$ be an $\IDA$ pattern such that $\pi_3$ visits $q'\in F$. Split
  $\pi_3$ in $q'\in F$, i.e. let $\pi_3=\pi_3^1\pi_3^2$ such that
  $\op{trg}(\pi_3^1)=\op{src}(\pi_3^2)=q'$ and let $x=l(\pi_3^1), y=l(\pi_3^2)$. Observe that $v=xy$ and
  let $\hat{v}=yx$. Now split the path $\pi_1$ in the same way, i.e.\
  $\pi_1=\pi_1^1\pi_1^2$ with $l(\pi_1^1)=x,l(\pi_1^2)=y$ and let $p'$ denote the state
  $\op{trg}(\pi_1^1)$. Observe that we now have three paths $\hat{\pi}_1 = \pi_1^2\pi_1\pi_1^1\in
  P(p',\hat{v}\hat{v},p'), \hat{\pi}_2 = \pi_1^2\pi_2\pi_3^1\in P(p', \hat{v}\hat{v},q')$
  and $\hat{\pi}_3 = \pi_3^2\pi_3\pi_3^1\in P(q',\hat{v}\hat{v},q')$, hence
  $(p',q',\hat{v}\hat{v},\hat{\pi}_{1,2,3})$ is an $\IDA$ pattern satisfying $\IDA_F$.

  $(\Rightarrow)$ for $(2)$: Let $(p,v,\pi_{1,2})$ be an $\EDA$ pattern that w.l.o.g.\
  visits $p'\in F$ on cycle $\pi_1$. Now let $\pi_1=\pi_1^1\pi_1^2$ such that
  $\op{trg}(\pi_1^1)=\op{src}(\pi_1^2)=q$ and let $x=l(\pi_1^1)$ and $y=l(\pi_1^2)$. Observe that $v=xy$
  and let $\hat{v}=yx$. Now we can take two different paths
  $\hat{\pi}_1=\pi_1^2\pi_1\pi_1^1,\hat{\pi}_2=\pi_1^2\pi_2\pi_1^1 \in
  P(p',\hat{v}\hat{v},p')$, hence $(p',\hat{v}\hat{v},\hat{\pi}_{1,2})$ is an $\EDA$
  pattern satisfying $\EDA_F$.
  \qed
\end{proof}

\lemedafimpliesidaf*
\begin{proof}
  Similar argument as for $\EDA \Rightarrow \IDA$ in \cite{weber1991degree}.
  Assume that $(r, v, \pi_{1,2})$ is an $\EDA$ pattern satisfying $\EDA_F$, i.e., $r\in F$.
  As $\pi_1\neq\pi_2$, let $v=xy$ such that the paths differ after reading $x$
  and let $\hat{v}=yx$.
  Take the first differing state on each of them, call them $p$ and $q$ and split the
  paths at this position into $\pi_1=\pi_1^1\pi_1^2,\pi_2=\pi_2^1\pi_2^2$. Now let
  $\hat{\pi}_1=\pi_1^2\pi_1\pi_1^1, \hat{\pi}_2=\pi_1^2\pi_2\pi_2^1,
  \hat{\pi}_3=\pi_2^2\pi_2\pi_2^1$. Notice that each of them visits $r\in F$, hence
  $(p,q,\hat{v}\hat{v}, \hat{\pi}_{1,2,3})$ is an $\IDA$ pattern
  and $r$ is visited on the path $\hat{\pi}_3$. This is equivalent to $\IDA_F$ by
  \Cref{lem:shifting}.
  \qed
\end{proof}

\lemnoedaboundamb*
\begin{proof}
  Let $q\in F$ and $w\in \Sigma^\omega$ such that there are $|Q|+1$ different
  infinite paths of $\mc{A}[q]$ visiting $q$ infinitely often. Now pick a time $i$ after
  which those infinite paths are separated and observe that now there must be two
  different prefixes $\pi_1$ and $\pi_2$ of these infinite paths that are in the same
  state $p\in Q$. Clearly, we can reach $q$ from $p$ again by following one of those two
  runs to the next visit of $q$ along some path $\pi_{pq}$. But then $\hat{\pi}_1 :=
  \pi_1\pi_{pq}$ and $\hat{\pi}_2 := \pi_2\pi_{pq}$ are two different paths from $q$ to
  $q$ and are labelled by the same prefix $\Pref_j(w)$ for some $j>i$ and hence $(q,
  \Pref_j(w), \hat{\pi}_{1,2})$ is an $\EDA$ pattern satisfying $\EDA_F$.
  \qed
\end{proof}

\lemidafc*
\begin{proof}
  For one direction,
  let $(p,q,v,\pi_{1,2,3})$ be an $\IDA$ pattern satisfying $\IDA_F$.
  Let $u\in \Sigma^*,\pi_0\in P(Q_0,u,p)$. Clearly, $uv^\omega\in L(\mc{A}_\NBA)$.
  Observe that for each $i\in \Nat$, $\pi_1$ can be taken $i$ times
  before using $\pi_2$ and then taking path $\pi_3$ forever. Hence,
  $\da(\mc{A}_\NBA)=\da(\mc{A}_\NBA, uv^\omega) \geq \aleph_0$.

  For the other direction, let $w\in L(\mc{A}_\NBA)$ and assume $\IDA_F$ does not hold,
  which implies $\lnot\EDA_F$ by \Cref{lem:edaf_implies_idaf}. Then by
  \Cref{lem:noeda_boundamb} a run can separate into at most $|Q|$ different accepting runs
  that visit $q\in F$ infinitely often, after seeing $q$ the first time.

  For contradiction, assume that there is some state $q\in F$ such that there are
  infinitely many accepting runs that visit $q$ infinitely often. This requires that the
  first visit of $q$ by a run can be delayed for an arbitrarily long time. But then there
  must exist an infinite path $\nu$ that never visits $q$, but from which infinitely many
  accepting runs can separate that visit $q$ infinitely often. Let $\alpha_i$ denote an
  accepting run that separated from $\nu$ at time $i$.

  Notice that whenever two different runs $\alpha_i$ and $\alpha_j$ with $i<j$ meet after $j$
  in the same state, they can be continued in the same way. Consider now only such
  modified runs $\hat{\alpha}_i$, so we can assume that if $\hat{\alpha}_i$ and
  $\hat{\alpha}_j$ with $i<j$ are in different states at time $k>j$, they have not met yet
  after $j$. Let $S_a^i$ be the set of different $\hat{\alpha}_j$ with $j\leq i$ that are in
  state $a$ at time $i$.

  Fix some time $i$ and pick $j>i$ such that at least $|Q|+1$ new runs $\hat{\alpha}_k$
  with $i<k<j$ separated from $\nu$. Observe that either all new runs joined existing sets
  directly or at least two sets of runs
  previously occupying different states must have met and joined before $j$, i.e. $S_a^i
  \cupdot S_b^i \subseteq S_c^j$ for some states $a,b,c$ such that $a\neq b$. But then there
  is an infinite sequence $S_{a_1}^1 \subseteq S_{a_2}^2 \subseteq \ldots$ with $a_i\in Q$
  that has a strict subsequence, as at most $|Q|$ different sets can exist at any
  time. Let $\hat{S}_i$ denote the set of such a sequence at time $i$.
  Each $\hat{\alpha}_i$ with $i\leq k$ that is in $\hat{S}_k$ hence also joins
  infinitely many $\hat{\alpha}_j$ with $j>i$ at some time $l>j$.

  Pick some state $p$ visited by $\nu$ infinitely often and some run $\hat{\alpha}_k$ from
  $\hat{S}_m$ with $k\leq m$. Observe that $\hat{\alpha}_k$ has at least one state $r$
  that it visits infinitely often at the same time as $\nu$ visits $p$. Pick a time $i>k$
  when this happens, i.e. $\nu(i) = p$ and $\hat{\alpha}_k(i) = r$. Now pick a run
  $\hat{\alpha}_l$ from $\hat{S}_{m'}$ such that $i<l<m'$. Finally, pick a time $j>l$ such
  that $\hat{\alpha}_k$ and $\hat{\alpha}_l$ have joined, then visited $q$ together at
  least once and finally $\nu(j)=p$ while $\hat{\alpha}_k(j)=\hat{\alpha}_l(j)=r$.

  Notice that all runs read the same finite substring $x$ of $w$ between $i$ and $j$.
  Further, $\nu$ witnesses a cycle $p\overset{x}{\rightarrow} p$, $\hat{\alpha}_k$
  witnesses a cycle $r \overset{x}{\rightarrow} r$ on which $q\in F$ is visited and
  finally, as $\hat{\alpha}_l$ separates from $\nu$ after $i$ and joins $\hat{\alpha}_k$
  before $j$, it witnesses a path $p \overset{x}{\rightarrow} r$. But this implies $\IDA_F$
  by \Cref{lem:shifting}, violating the assumption.

  Hence for each $q\in F$ there is a finite time after which all accepting runs
  that will visit $q$ infinitely often must have visited $q$ at least once.
  But then the number of accepting runs on $w$ must be finite.
  \qed
\end{proof}

\clearpage
\subsection{Proof of $\op{PSPACE}$-completeness (\Cref{thm:ambpspace})}

We will make use of the following two observations:

\begin{lemma}
  \label{lem:ultper}
  If $\da(\mc{A}_\NBA, w) = k$, then there exists an ultimately periodic word
  $w'=xy^\omega \in L(\mc{A}_\NBA)$ such that $\da(\mc{A}_\NBA,w') \geq k$.
\end{lemma}
\begin{proof}
  Let $w\in L(\mc{A})$ with $\da(\mc{A},w)=k$.
  Consider a sequence $t_0,t_1,\ldots$ of ordered tuples $t_i\in Q^k$ of
  states of all $k$ different accepting infinite runs at time $i$ and observe that as
  the number of different tuples is finite, there is an infinite sequence $i_0 < i_1 <
  \ldots$ such that all $t_{i_j}$ are equal. Now let
  $i_a$ be the first time in this sequence such that all $k$ runs already have separated
  and let $i_b > i_a$ be the first time after $i_a$ such that each run has visited an
  accepting state again. Then let $x$ be the prefix of $w$ that was read up to $i_a$ and
  $y$ the substring that was read between $i_a$ and $i_b$. By construction the
  ultimately periodic word $w'=xy^\omega$ has at least as many accepting runs as $w$.
  \qed
\end{proof}

\begin{lemma}
  \label{lem:uniqpaths}
  If $\mc{A}$ is a finitely ambiguous $\NBA$, then in each SCC $C$, $P(p,x,q)$
  contains at most one path for every pair of states $p,q\in C$ and finite string $x$.
\end{lemma}
\begin{proof}
  Assume there are $\pi_1,\pi_2\in P(p,x,q),\pi_1\neq \pi_2$. As we are in an
  SCC, there is a path $\pi_{qp}\in P(q,y,p)$. But then we have two different cycles
  $\hat{\pi_1}=\pi_1\pi_{qp}$ and $\hat{\pi_2}=\pi_2\pi_{qp},
  \hat{\pi_1},\hat{\pi_2}\in P(p,xy,p)$, which implies $\EDA$.
  \qed
\end{proof}

Now we can adapt the algorithm presented in \cite{chan1988finite} to show
$\op{PSPACE}$-completeness of this problem from $\NFA$ to $\NBA$, proving the following
result:
\thmambpspace*
\begin{proof}
  First, we show that the problem can be decided in $\op{PSPACE}$.

  Let $\mc{A}$ be a finitely ambiguous $\NBA$ with states $Q=\{q_1,\ldots,q_n\}$ and
  w.l.o.g. assume that $q_1$ is the initial state.

  Define a set of matrices $\{T_a\}_{a\in \Sigma}$ with
  $T_a(i,j) = 1$ if $q_j\in \Delta(q_i,a)$ and $0$ otherwise.
  For each finite word $w\in \Sigma^*$, let $T_w=T_{w(1)}\cdot\ldots\cdot T_{w(n)}$.
  Then $T_w(i,j)$ is the number of different paths from $q_i$ to $q_j$ labelled with $w$
  and hence $T_w(1,k)$ denotes the number of different runs that are in state $q_k$ after
  reading $w$.

  Now let $C$ be an SCC of $\mc{A}$ and let $\pi_C(A)$ denote the
  restriction of matrix $A$ to the rows and columns that correspond to states in
  $C$. Then clearly the matrix $\pi_{C}(T_w)$ describes the number of
  different $w$-labelled paths between states of $C$. Remember that by
  \Cref{lem:uniqpaths} a finitely ambigious $\NBA$ has at most one path for each finite word
  between each pair of states that are within the same SCC, so $\pi_{C}(T_w)$ contains
  only $0$ and $1$ as values.

  By \Cref{lem:ultper} we can restrict ourselves to words of the form $xy^\omega \in
  L(\mc{A}_\NBA)$ and clearly we can choose $xy^\omega$ such that after reading $x$ all
  accepting runs are already separated and have reached their terminal SCC which is never left
  again. By choice of $x$ and \Cref{lem:uniqpaths} all accepting runs are continued
  unambiguously on $y^\omega$ after reading $x$.

  So to obtain the number of accepting runs, we need to sum up $T_x(1,k)$ for all states
  $q_k$ that lie on an accepting cycle when reading $y^\omega$. As there is only a finite
  number of such unambiguous continuation cycles, after a finite number of iterations
  of $y$ (least common multiple of all different cycle lengths and $|y|$) all runs are
  back in the state of the cycle where they started. So let $z=y^j$ for some $j\in\Nat$
  such that $T_{z}(k,k)=1$ for every cycle from some $q_k$ back to $q_k$.

  What is left to verify is that such a cycle is accepting. To do this,
  we define another set of matrices, $\{A_a\}_{a\in\Sigma}$ with
  $A_a(i,j) = 2$ if $q_j\in \Delta(q_i,a),q_j\in F$, $A_a(i,j) = 1$ if
  $q_j\in \Delta(q_i,a),q_j\not\in F$ and $0$ otherwise.

  Now $A_z(i,i)>1$ iff the unique $z$-labelled cycle from $q_i$ to $q_i$ visits
  at least one accepting state. Using this, we obtain the following result:
  $\da(\mc{A}) = \da(\mc{A},xy^\omega) = \Sigma_{i\in I} T_x(1,i)$ with $I = \{i \mid
  T_z(i,i)=1 \land A_z(i,i) > 1 \}$ for some choice of $x,y\in \Sigma^*, j\in
  \Nat$ and $z=y^j$.

  Hence, the following nondeterministic polynomial space algorithm decides whether
  $\da(\mc{A})>d$, by guessing a prefix $x$ that yields candidate paths
  and guessing $z$ to identify paths that are prefixes of an accepting run on
  the same word, using the reasoning above.

  \begin{algorithmic}
    \State{$X := Id, Z := Id, \hat{Z} := Id$}
    \Loop
    \State{guess $a\in \Sigma$}
    \State{guess $b\in \{0,1\}$}
    \If{b=0}
    \State{$X := X\otimes_{d+1} T_a$}
    \Else
    \State{$Z := Z\otimes_{2} T_a$}
    \State{$\hat{Z} := \hat{Z}\otimes_{2} A_a$}
    \EndIf

    \State{$I := \{i \mid Z(i,i)=1 \land \hat{Z}(i,i) > 1 \}$}
      \If{$\Sigma_{i\in I} X(1,i) > d$}
      \State{accept, halt}
      \EndIf
    \EndLoop
  \end{algorithmic}

  where $Id$ is the identity matrix and $\otimes_n$ is a matrix multiplication that
  identifies all integers $>n$ with $n$ (notice that this ensures the bounded space
  usage of the algorithm).

  Now, we will show completeness. Deciding the ambiguity of a finitely ambiguous $\NFA$ is
  $\op{PSPACE}$-complete by \cite{chan1988finite}. We perform a reduction of the
  corresponding $\NFA$ problem to the $\NBA$ variant by introducing a fresh symbol $\#$
  and fresh state $q_\#$. The $\NBA\ \mc{A}'$ is defined by $\Sigma'=\Sigma\cup\{\#\}, Q'
  = Q\cup\{q_\#\}, Q_0'=Q_0, F'=\{q_\#\}$ and $\Delta' = \Delta\cup\{(q,\#,q_\#)\mid
  q=q_\# \lor q\in F \}$. It is easy to see that $L(\mc{A}'_\NBA) = \{ w\#^\omega \mid w \in
  L(\mc{A}_\NFA)\}$, i.e.\ each word of the $\NBA$ $\mc{A}$' is in one-to-one correspondence
  with a word of the $\NFA$ $\mc{A}$ that has the same number of accepting runs, as the only
  accepting SCC consists of $\{q_\#\}$, which is trivially unambiguous. Hence we decide
  $\da(\mc{A})>d$ by deciding $\da(\mc{A}')>d$ using the presented algorithm.
  \qed
\end{proof}

\clearpage
\subsection{Proof of the lower bound (\Cref{thm:fabalowerbound})}

\thmfabalowerbound*
\begin{proof}
  In \cite{leung1998separating} this result is presented for the translation of $\NFA$ to
  polynomially ambiguous $\NFA$ and a corresponding family $\mc{A}_n$ of worst-case $\NFA$
  of size $n$ is provided. Clearly translation to finitely ambiguous automata is not
  easier.

  Pick some $n>0$.
  Let $\NBA\ \hat{A}_{n}$ be defined by extending $\NFA$ $\mc{A}_{n}$ with a fresh symbol
  $\#$ and additional transitions $\{(q_F, \#, q_0) \mid q_F \in F, q_0 \in Q_0 \}$.
  Let $\hat{\mc{B}}$ be a finitely ambiguous and trim $\NBA$ accepting the same language.
  We obtain an $\NFA\ \mc{B}$ from $\hat{\mc{B}}$ by defining only the states with an
  outgoing $\#$-labelled transition as accepting and then removing those $\#$-labelled
  transitions. This $\NFA$ is also finitely ambiguous by \Cref{lem:nfanba}, as removing
  transitions clearly can not increase ambiguity.

  Let $w \in L(\mc{A}_n)$. Then by construction, at least the infinite word $(w\#)^\omega$
  is in $L(\hat{\mc{A}}_n)$ and hence in $L(\hat{\mc{B}})$, so a $\#$-labelled transition
  must be possible after reading the prefix $w$ on some accepting run in $\hat{\mc{B}}$.
  By definition then we have $w \in L(\mc{B})$.

  Let $w \not\in L(\mc{A}_n)$. Then by construction there is no infinite word with prefix
  $w\#$ in $L(\hat{\mc{A}}_n)$ and hence neither in $L(\hat{\mc{B}})$. As $\hat{\mc{B}}$
  is trim, a $\#$-labelled transition is not possible on any run after reading prefix $w$,
  and then by definition we have $w \not\in L(\mc{B})$.

  So we have $L(\mc{B}) = L(\mc{A})$, which implies that  $\mc{B}$ must have at least
  $2^{n}-1$ states by \cite{leung1998separating} and hence $\hat{\mc{B}}$ as well, by
  construction.
  \qed
\end{proof}

\end{document}